\documentclass{article}

\usepackage[english]{babel}
\usepackage[a4paper, top=3cm, bottom=3cm, left=3cm, right=3cm]
{geometry}
\usepackage{graphicx}
\usepackage{array} 
\usepackage{adjustbox}
\usepackage{tabularx}
\usepackage{amsmath,amssymb,amsthm,amsfonts,mathtools}
\usepackage{booktabs}
\usepackage{graphicx}
\usepackage{xcolor}
\usepackage{multirow}
\usepackage{multicol}
\usepackage{bigstrut}
\usepackage{algorithmic}
\usepackage[ruled]{algorithm2e}
\usepackage[colorlinks=true, allcolors=blue]{hyperref}
\usepackage{mfirstuc} 
\usepackage{float}
\usepackage{makecell}
\usepackage{appendix}
\usepackage{arydshln}
\usepackage{soul}

\newcolumntype{P}[1]{>{\centering\arraybackslash}p{#1}}

\newtheorem{proposition}{Proposition}
\newtheorem{theorem}{Theorem}

\newtheorem{lemma}{Lemma}
\newtheorem{definition}{Definition}

\newtheorem{claim}{Claim}

\newcommand{\OPTESW}{\text{OPT-ESW}}
\newcommand{\ESW}{\text{ESW}}
\newcommand{\maxESW}{\max_{\mathcal{A}\in \mathcal{C}_\kappa(I)}\ESW}
\newcommand{\maxUSW}{\max_{\mathcal{A}\in \mathcal{C}_\kappa(I)}\USW}
\newcommand{\maxUSWa}{\max_{\mathcal{A}\in \mathcal{C}_\kappa(I')}\USW}
\newcommand{\maxUSWb}{\max_{\mathcal{A}\in \mathcal{C}_\kappa(I'')}\USW}
\newcommand{\OPTUSW}{\text{OPT-USW}}
\newcommand{\USW}{\text{USW}}

\newcommand{\cA}{\mathcal{A}}
\newcommand{\cB}{\mathcal{B}}

\DeclarePairedDelimiter\floor{\lfloor}{\rfloor}
\usepackage{xcolor}

\begin{document}
\title{The (Exact) Price of Cardinality for Indivisible Goods: A Parametric Perspective\thanks{The author names are ordered alphabetically.}}

\author{Alexander Lam\thanks{Department of Computer Science, City University of Hong Kong. alexlam@cityu.edu.hk} \and Bo Li\thanks{Department of Computing, The Hong Kong Polytechnic University. comp-bo.li@polyu.edu.hk} \and
Ankang Sun\thanks{Department of Computing, The Hong Kong Polytechnic University. ankang.sun@polyu.edu.hk} }

\date{}
\maketitle   

\begin{abstract}
We adopt a parametric approach to analyze the worst-case degradation in social welfare when the allocation of indivisible goods is constrained to be \emph{fair}.
Specifically, we are concerned with \emph{cardinality}-constrained allocations, which require that each agent has at most $k$ items in their allocated bundle.
We propose the notion of the \emph{price of cardinality}, which captures the worst-case multiplicative loss of \emph{utilitarian} or \emph{egalitarian} social welfare resulting from imposing the cardinality constraint. We then characterize tight or almost-tight bounds on the price of cardinality as exact functions of the instance parameters, demonstrating how the social welfare improves as $k$ is increased.
In particular, one of our main results refines and generalizes the existing asymptotic bound on the price of balancedness, as studied by Bei et al. \cite{DBLP:journals/mst/BeiLMS21}.
We also further extend our analysis to the problem where the items are partitioned into disjoint categories, and each category has its own cardinality constraint. Through a parametric study of the price of cardinality, we provide a framework which aids decision makers in choosing an ideal level of cardinality-based fairness, using their knowledge of the potential loss of utilitarian and egalitarian social welfare. 

\end{abstract}

\section{Introduction}

The allocation of indivisible goods to a set of agents is a ubiquitous problem in our society, capturing a number of real-world scenarios. 
For example, an inheritance may involve indivisible goods such as jewelry, cars, and estates, and food banks are constantly faced with the task of allocating donations to people in need.
In a corporate setting, equipment and human resources such as developers and designers need to be assigned to various projects and departments. These real-world scenarios often involve constraints which are imposed on the allocation, which may make the allocation difficult to compute, or prevent it from being socially optimal. 

The most commonly studied constraint in resource allocation is the requirement that the allocation should be \emph{fair} to the agents \cite{DBLP:journals/ai/AmanatidisABFLMVW23}. Fairness constraints can be expressed in several different ways: a \emph{proportional} allocation guarantees that each of the $n$ agents receives at least $\frac{1}{n}$th of their utility for the entire set of items, and in a \emph{balanced} allocation, the goods are spread out among the agents as evenly as possible, so that the number of items received by each agent differs by at most one. 
The latter notion of \emph{balancedness} is a natural constraint which may be imposed by the central decision maker due to its simplicity and ease of implementation, without requiring knowledge of the agents' utilities.
However, this constraint may severely degrade the social welfare of the allocation.
This is particularly true in instances where agents have low utility for a large number of items, rather than highly valuing a small subset of items, motivating the need for a weaker, more variable notion of `cardinal fairness'.

The main focus of our paper is on the constraint of \emph{cardinality}, a generalization of balancedness which imposes an upper limit of $k$ on the number of items an agent may receive from the allocation.\footnote{Note that the cardinality constraint is equivalent to balancedness when $m\in \{kn-k+1,kn-k+2,\dots,kn\}$, where $m$ is the total number of items.}
The cardinality constraint is commonly applied in practice. For example, when a university provides funding to support PhD students, it is typical to impose a limit on the number of students each professor can supervise, due to fairness concerns. 
Intuitively, the potential loss of welfare decreases as $k$ increases. The central decision maker may vary the parameter of $k$ to achieve their desired tradeoff between the level of balancedness and the social welfare of the cardinality-constrained allocation. More generally, the items may also be partitioned into disjoint categories, where each category $j$ has its own cardinality constraint of $k_j$. However, the exact effect of the value of $k$ on the social welfare objectives remains unclear, leading to the following research question which we address in this paper:

\begin{quote}
    \emph{What is the worst-case (multiplicative) loss of social welfare when there is a limit on the number of items each agent can receive in an allocation, and how does the loss change as the limit is varied?}
\end{quote}

In particular, we aim to quantify this loss in an \emph{exact} sense, as opposed to the asymptotic bounds which are common in the literature. This helps with making a more informed decision on the cardinality constraint values, particularly in scenarios where the number of agents and/or items is small.

\subsection{Our Contribution}
In this work, we initiate the study of \emph{price of cardinality} from the parametric perspective.
We define the price of cardinality as the worst-case ratio between the welfare of the optimal allocation, and the optimal welfare among all cardinality-constrained allocations.
Our work concerns both \emph{utilitarian} and \emph{egalitarian} social welfare, defined as the sum of agents' utilities and worst-off agent's utility respectively.
For both objectives, we establish tight or almost-tight bounds for the price of cardinality in both the single-category and multi-category cases, expressing the prices as exact functions of the cardinality parameters, as opposed to the common asymptotic bounds in the literature. 
A benefit of our parametrized approach is that it enables decision makers to choose their desired level of fairness based on the potential loss of social welfare. We summarize our main results as follows.

\paragraph{Single category.} 
We start with the single category case, where each agent can receive at most $k$ items. 
We show that for any instance with $n$ agents and $m$ items such that $k\ge \frac{m}{n}$, the utilitarian price of cardinality is $\frac{1}{2}\left(1+\sqrt{1+\frac{m-1}{k}}\right)$. This can be visualized in Figure~\ref{fig:util} as a function of the ratio $\frac{k}{m-1}$. 
We note that this bound is precisely tight for instances where $m$ and $k$ satisfy a divisibility constraint, and tight up to an additive constant that is smaller than 1 for every instance. Furthermore, when $m=kn$, this result coincides with (and refines) the asymptotic bound of $\Theta(\sqrt{n})$ for the \emph{price of balancedness} \cite{DBLP:journals/mst/BeiLMS21}.

For the objective of egalitarian social welfare, we present an exact bound of $\max\{ \frac{m-n+1}{k}, 1\}$ which is tight for all instances. This result shows that when $k$ is small compared to $m$, the cardinality constraint may adversely affect the egalitarian fairness of the allocation, particularly when the allocation is constrained to be balanced (i.e., $m=kn$).

\begin{figure}
    \centering
    \includegraphics[width=0.7\linewidth]{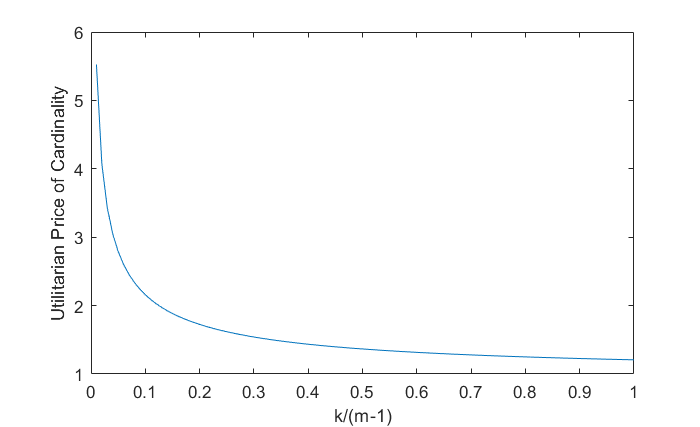}
    \caption{Plot of the utilitarian price of cardinality in the single-category setting as a function of $\frac{k}{m-1}$.}
    \label{fig:util}
\end{figure}

\paragraph{Multiple categories.}
For utilitarian social welfare, we first focus on the case of two agents, giving an exact and tight bound of $\frac{2m_1m_2}{m_2k_1+m_1k_2}$ for the utilitarian price of cardinality.\footnote{We order the categories such that $\frac{k_1}{m_1}\leq \frac{k_2}{m_2}\leq \dots \leq \frac{k_h}{m_h}$, where $k_j$ and $m_j$ denote the cardinality constraint and number of items in category $j$, respectively.} For the case of general $n$ agents, we establish a utilitarian price of cardinality of $\frac{m_1}{k_1}$, which is tight for instances where there are $n$ categories and $\frac{k_1}{m_1}=\dots=\frac{k_n}{m_n}$.

Finally, in the multi-category case, we establish an exact bound for egalitarian price of cardinality as a function of $n$, the cardinality constraints $k_j$, and the number of items in each category $m_j$, and this bound is tight for all instances.

For all of our main results, we also describe how a cardinal allocation with guaranteed welfare corresponding to the respective price of cardinality can be found.

\subsection{Related Work}

The problem of allocating a set of indivisible goods to self-interested agents has been extensively studied in computer science and economics in recent years, with a primary focus on finding allocations which are constrained to be `fair' in some axiomatic sense (see, e.g., the seminal papers by Gamow and Stern \cite{GS58}, Steinhaus \cite{steinhaus1948problem}, Varian \cite{Varian74} and the surveys by Amanatidis et al. \cite{DBLP:journals/ai/AmanatidisABFLMVW23}, Walsh \cite{Wals20}).
Common fairness notions have been shown to be compatible with cardinality constraints: for example, Biswas and Barman \cite{DBLP:conf/ijcai/BiswasB18} and Hummel and Hetland \cite{DBLP:conf/atal/HummelH22} considered how to compute EF1 and approximate MMS allocations under cardinality constraints.

While most related work is concerned with computing fair allocations, there is a growing body of research on quantifying the degradation of welfare when fairness constraints are imposed. This was first proposed by Caragiannis et al. \cite{DBLP:journals/mst/CaragiannisKKK12}, and results on the price of fairness have since been extended by Barman et al. \cite{DBLP:conf/wine/BarmanB020}, Li et al. \cite{DBLP:conf/www/0037L022}, and Li et al. \cite{DBLP:conf/atal/LiLLTT24}. These papers generally examine the effect of fairness constraints from an asymptotic perspective, e.g., the prices of EF1 and $\frac{1}{2}$-MMS are $\Theta(\sqrt{n})$, and since computing the exact bound can be challenging for every value of $n$, the literature often restricts to the case where $n=2$. In particular, the paper by Bei et al. \cite{DBLP:journals/mst/BeiLMS21} gives an asymptotic bound of $\Theta(\sqrt{n})$ (and a precise bound of $4/3$ for $n=2$) for the price of \emph{balancedness} for utilitarian social welfare. 
The price of fairness has also been studied in the allocations of indivisible chores \cite{DBLP:journals/aamas/SunCD23a,DBLP:journals/aamas/SunCD23,DBLP:conf/atal/Sun024}.
As asymptotic bounds may obscure the precise effects of constraints (particularly in scenarios with a small number of agents), this paper aims to generalize the existing asymptotic bound on the price of balancedness to the exact bound on the price of cardinality, and extend the results to multiple categories.

Beyond cardinality constraints, other types of constraints such as
connectivity
\cite{DBLP:conf/ijcai/BouveretCEIP17,DBLP:journals/corr/abs-2405-03467,DBLP:conf/ecai/SunL23},
geometric \cite{SNHA17}, and separation \cite{ESS22} have also been studied in the context of fair division.
In particular, the cardinality constraint in the single-category case is equivalent to budget-feasibility (as studied by Wu et al. \cite{DBLP:conf/ijcai/00010G21}) when agents have identical budgets and items have identical costs.
For a recent overview on constraints in fair division, we refer the readers to the survey by Suksompong \cite{DBLP:journals/sigecom/Suksompong21}. 
\section{Preliminaries}

For any $k\in \mathbb{N}^+$, denote $[k]\coloneqq \{1,\ldots,k\}$. We have a set $N$ of $n$ agents who are to receive a set $M=\{g_1,\ldots,g_m\}$ of $m$ indivisible goods. Each agent $i\in N$ has a utility function $u_i : 2^M\rightarrow \mathbb{R}_{\geq 0}$ over each possible subset of goods. Let $\mathcal{U} = (u_1, u_2, \dots, u_n)$ denote the agents' utility profiles.
Throughout the paper, we assume the utilities are \emph{additive} (i.e., $u_i(S)=\sum_{g\in S}u_i(\{g\})$) and \emph{normalized} (i.e., $u_i(\emptyset)=0$ and  $u_i(M)=1$ for each $i\in N$). We slightly abuse notation and write $u_i(g)$ instead of $u_i(\{g\})$.

An allocation $\mathcal{A}$ is a partition of the items into $n$ disjoint bundles $\mathcal{A}=(A_1,\dots,A_n)$, with agent $i$ receiving $A_i$. The set of goods is partitioned into $h$ different non-overlapping categories $C=\{C_1,\dots,C_h\}$, with respective cardinality constraints $k_1,\dots,k_h$. 
The cardinality constraint $k_j$ specifies the maximum number of items from category $C_j$ that each agent can have under a \emph{cardinal} allocation. In other words, an allocation $\mathcal{A}$ is \emph{cardinal} if for each bundle $A_i$ and category $C_j$, $|A_i \cap C_j|\leq k_j$ holds. 
Our prices of cardinality will be expressed in terms of the number of agents $n$, the cardinality constraints $k_j$, and the number of items in each category, so for simplicity, we denote, for each $j\in [h]$, $m_j\coloneqq |C_j|$ and $m = \sum_j m_j$. We also assume that for each $j\in [h]$, $k_j\geq \frac{m_j}{n}$ so that every item can be allocated.
In the single-category scenario (i.e., $h=1$), we slightly abuse notation and use $m$ and $k$ instead of $m_1$ and $k_1$ to refer to the number of items and the cardinality constraint, respectively.

We refer to a problem setting with agents $N$, goods $M$, utility profile $\mathcal{U}$ and partition of goods into categories $C$ as an \emph{instance}, denoted as $I=\langle N,M,\mathcal{U},C\rangle$.  

Our results are concerned with the following objectives of \emph{utilitarian social welfare} and \emph{egalitarian social welfare}.

\begin{definition}
Given an instance $I$ and an allocation $\mathcal{A}=(A_1,\dots,A_n)$ of the instance,
\begin{itemize}
    \item the \emph{utilitarian social welfare}, or the sum of the agents' utilities, is denoted by $\USW(\mathcal{A})\coloneqq \sum_{i\in N}u_i(A_i)$.
    \item  the \emph{egalitarian social welfare}, or the utility of the worst-off agent, is denoted by $\ESW(\mathcal{A})\coloneqq \min_{i\in N}u_i(A_i)$.
\end{itemize}
\end{definition}
We denote the optimal utilitarian (resp. egalitarian) social welfare over all possible allocations of an instance $I$ as $\OPTUSW(I)$ (resp. $\OPTESW(I)$).

Given an instance $I$ and cardinality constraints $\kappa=(k_1,\dots,k_h)$, let the set of all cardinal allocations be denoted as $\mathcal{C}_\kappa(I)$. We now define the main concept of our paper: the \emph{price of cardinality}, for our objective functions of utilitarian and egalitarian social welfare.

\begin{definition}\label{def:util:pcar}
    The \emph{utilitarian price of cardinality} is defined as
    $$\sup_{I=\langle N,M,\mathcal{U}, C\rangle, \kappa=(k_1,\dots,k_h)} \frac{\OPTUSW(I)}{\max_{\mathcal{A}\in \mathcal{C}_\kappa(I)}\USW(\mathcal{A})}.$$
\end{definition}

\begin{definition}
    The \emph{egalitarian price of cardinality} is defined as
    $$\sup_{I=\langle N,M,\mathcal{U}, C\rangle, \kappa=(k_1,\dots,k_h)} \frac{\OPTESW(I)}{\max_{\mathcal{A}\in \mathcal{C}_\kappa(I)}\ESW(\mathcal{A})}.$$
\end{definition}
For an instance $I$, if $\OPTESW(I)=0$, we define that $\frac{0}{0}=1$ and the egalitarian price of cardinality to be $1$.

\section{Single Category}\label{sec::single-cat}
We first give results for the case where there is only one category, and therefore only one cardinality constraint $k$. Note that if $m\leq k$, then the cardinality constraint will have no effect, and the price of cardinality will be equal to $1$. We therefore assume that $m>k$ throughout this section.

\subsection{Utilitarian Social Welfare}
We begin with the utilitarian price of connectivity in the single category setting. Interestingly, the price of cardinality for utilitarian welfare is only directly dependent on the number of items and the cardinality constraint, and not directly dependent on the number of agents (it is not entirely independent of $n$ as we require $n\geq \frac{m}{k}$).

\begin{theorem}\label{thm:pocusw1-ak}
    In the single-category case, 
    the utilitarian price of cardinality is $$\frac{1}{2}\left(1+\sqrt{1+\frac{m-1}{k}}\right).$$
\end{theorem}
\begin{proof}
    The lower bound will be proven in Lemma~\ref{lem:1catuswlower}, and the upper bound will be proven in Lemmas~\ref{lem:fix} and~\ref{lem:1catuppercase2}.
\end{proof}
It is worth noting that the stated price of cardinality expression is tight in a very strong sense:
\begin{itemize}
    \item For any instance with cardinality constraint $k$ and $m$ items, the utilitarian price of cardinality is \emph{at most} $\frac{1}{2}\left(1+\sqrt{1+\frac{m-1}{k}}\right)$.
    \item As we will shortly show in Lemma~\ref{lem:1catuswlower}, for any instance with cardinality constraint $k$ and $m=k(c^2-1)+1$ items, where $c\in \mathbb{N}^+\setminus\{1\}$, the utilitarian price of cardinality is exactly $\frac{1}{2}\left(1+\sqrt{1+\frac{m-1}{k}}\right)$.
    \item In Lemma~\ref{lem:1catuswlowergen}, we further show that for any other instance with cardinality constraint $k$, the utilitarian price of cardinality is at least $\frac{1}{2}\left(-1+\sqrt{1+\frac{m-1}{k}}\right)$.
\end{itemize}

This is due to the lower bound construction requiring $m$ and $k$ to meet a divisibility constraint. Specifically, there is a set of agents who each value the same item $g_m$ at one utility, and the remaining agents' utilities are such that they each receive the same number of items under the utilitarian-optimal allocation.
\begin{lemma}\label{lem:1catuswlower}
    In the single-category case, if $m=k(c^2-1)+1$ for some $c\in \mathbb{N}^+\setminus\{1\}$, then the utilitarian price of cardinality is at least $\frac{1}{2}\left(1+\sqrt{1+\frac{m-1}{k}}\right).$
\end{lemma}
\begin{proof}
    Let $I$ be an instance with cardinality constraint $k$ and $m=k(c^2-1)+1$ items for some $c\in \mathbb{N}^+\setminus\{1\}$, and let $s\coloneqq -1+\sqrt{1+\frac{m-1}{k}}$.
    Note that
    \begin{itemize}
        \item $s=c-1\in \mathbb{N}^+$,
        \item and $\frac{m-1}{s}=\frac{k(c^2-1)}{c-1}\in \mathbb{N}^+$.
    \end{itemize}
    
    The agents' utilities are as follows:
    \begin{itemize}
        \item for $i=1,\ldots,s$, if $j=(i-1)\frac{m-1}{s}+1,\ldots,i\frac{m-1}{s}$, then $u_i(g_j)=\frac{s}{m-1}$; otherwise, $u_i(g)=0$;
        \item for $i\geq s+1$, $u_i(g_m)=1$ and $u_i(g)=0$ for each $g\in M\setminus \{g_m\}$.
    \end{itemize}
    We have $\OPTUSW(I)=1+s$ as in the utilitarian-optimal allocation, each agent $i\in [s]$ receives utility $1$, and agents $s+1,\ldots,n$ have a total utility of $1$.
    We also have $\maxUSW(\mathcal{A})=1+\frac{ks^2}{m-1}$ by letting every agent $i\in [s]$ keep their $k$ most valued items; note $\frac{m-1}{s}>k$. 
    Dividing these terms and substituting $s=-1+\sqrt{1+\frac{m-1}{k}}$, we get our price of cardinality lower bound of
     $\frac{\OPTUSW(I)}{\maxUSW(\mathcal{A})}=\frac{1}{2}\left(1+\sqrt{1+\frac{m-1}{k}}\right)$.
\end{proof}

Note that if the divisibility constraint is not met by $m$ and $k$ (i.e., $m\neq k(c^2-1)+1$ for all $c\in \mathbb{N}^+\setminus\{1\}$), we can still construct a similar lower bound which will be slightly lower than our general upper bound (as exemplified in Figure~\ref{fig:gap}). We take $s=\floor{-1+\sqrt{1+\frac{m-1}{k}}}$, and let agents $1,\dots,s$ each equally value a disjoint subset of the items $g_1,\dots,g_{m-1}$, such that the number of goods they positively value differs by at most $1$. The gap between the upper and lower bound here is due to the rounding of $s$ and the partitioning of $m-1$ items among $s$ agents as evenly as possible. Furthermore, the following lemma implies that the price of cardinality lower bound corresponding to this instance construction differs from our stated upper bound by an additive constant of at most $1$.

\begin{lemma}\label{lem:1catuswlowergen}
    In the single-category case, if $m\neq k(c^2-1)+1$ for all $c\in \mathbb{N}^+\setminus\{1\}$, then the utilitarian price of cardinality is at least $\frac{1}{2}\left(-1+\sqrt{1+\frac{m-1}{k}}\right).$
\end{lemma}

\begin{figure}
    \centering
    \includegraphics[width=0.7\linewidth]{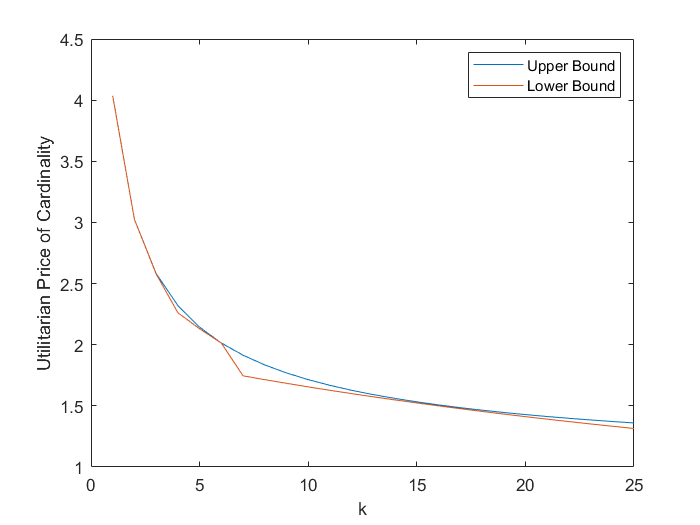}
    \caption{Plot for $m=50$ showing the gap between the lower bound as described in the main body and proof of Lemma~\ref{lem:1catuswlowergen} for any $m$ and $k$, and the upper bound from Theorem~\ref{thm:pocusw1-ak}.}
    \label{fig:gap}
\end{figure}

The missing proofs are deferred to the appendix. We now prove the upper bound of Theorem~\ref{thm:pocusw1-ak}, which holds for any $m$ and $k$. First, we fix an arbitrary instance $I$ with $m$ items and cardinality constraint $k$, and denote the utilitarian-optimal allocation\footnote{In the case of ties, the tiebreak procedure will be described in the upcoming proofs when necessary.} of $I$ by $\cA^*=(A^*_1,\ldots,A^*_n)$. We then define the following subsets of agents. Under $\cA^*$,
\begin{itemize}
    \item The set of agents receiving less than $k$ items is $R=\{ i\in[n] : |A^*_i| < k \}$.
    \item The set of agents receiving exactly $k$ items is $T=\{ i\in[n] : |A^*_i| = k \}$.
    \item The set of agents receiving more than $k$ items is $S=\{ i\in[n] : |A^*_i| > k \}$.
\end{itemize}
We can also assume that $\cA^*$ does not satisfy the cardinality constraint, and hence $S\neq \emptyset$, implying $R\neq \emptyset$.
We next divide the proof of the upper bound into two cases depending on whether $\sum_{i\in R\cup T}u_i(A^*_i)\geq 1$ or $\sum_{i\in R\cup T}u_i(A^*_i)< 1$. We begin with the former case.

\begin{lemma}\label{lem:fix}
    Let $I$ be a (possibly non-normalized) instance satisfying $0<u_i(M)\leq 1$ for each $i\in S$ and $\sum_{i\in N}\setminus Su_i(A^*_i)\geq 1$, where $S$ is the set of agents who receive more than $k$ items under utilitarian-optimal allocation $\cA^*=(A^*_1,\ldots,A^*_n)$. Then,
    $$\frac{\OPTUSW(I)}{\maxUSW(\mathcal{A})}\leq \frac{1+\sum_{i\in S}u_i(A^*_i)}{1+k\sum_{i\in S}\frac{u_i(A^*_i)}{|A^*_i|}}\leq \frac{1+s}{1+\frac{ks^2}{m-1}},$$ where $s=-1+\sqrt{1+\frac{m-1}{k}}$.
\end{lemma}
\begin{proof}[Proof Sketch]
    We present a proof sketch for the inequality $\frac{1+\sum_{i\in S}u_i(A^*_i)}{1+k\sum_{i\in S}\frac{u_i(A^*_i)}{|A^*_i|}}\leq \frac{1+s}{1+\frac{ks^2}{m-1}}.$ By taking derivatives, we find that the LHS is maximized when every $u_i(A^*_i)$ is either $1$ or $0$. This gives us
    \begin{align*}
        \frac{1+\sum_{i\in S}u_i(A^*_i)}{1+k\sum_{i\in S}\frac{u_i(A^*_i)}{|A^*_i|}}&\leq \frac{1+ \sum_{i \in S'} 1}{1+k\sum_{i \in S'}\frac{1}{A^*_i}} \\
        &\leq \frac{1+|S'|}{1+\frac{k|S'|^2}{m-1}}\leq \frac{1+s}{1+\frac{ks^2}{m-1}},
    \end{align*}
    where $s=-1+\sqrt{1+\frac{m-1}{k}}$. Here, the second inequality follows from the arithmetic mean-harmonic mean (AM-HM) inequality, and the final inequality follows from taking derivatives with respect to $|S'|$.
\end{proof}

We now address the remaining case where $\sum_{i\in R\cup T} u_i(A^*_i) < 1$. 
In this case, we assume that instance $I$ is \emph{preprocessed} such that for each $j\in R$, $
\sum_{i\in S}u_{j}(A^*_i):= 1-\sum_{i\in R\cup T}u_i(A^*_i).
$
As $\cA^*$ is the utilitarian-optimal allocation, we have $\sum_{i\in R\cup T}u_{i}(A^*_i)\geq \sum_{i\in R\cup T}u_{j}(A^*_i)$ for all $j\in R$, and therefore before processing, we have $\sum_{i\in S}u_j(A^*_i) \geq 1-\sum_{i\in R\cup T} u_i(A^*_i)$.
Accordingly, this preprocessing can be achieved by reducing the utility that each agent $j\in R$ has for items $\{A^*_i\}_{i\in S}$ until we reach $\sum_{i\in S}u_{j}(A^*_i)= 1-\sum_{i\in R\cup T}u_i(A^*_i)$.
The new/preprocessed instance is not necessarily normalized, meaning there may exist $i\in R$ such that $0\leq u_i(M)\leq 1$. 
Note that the optimal utilitarian welfare does not change before and after the preprocessing, but the optimal utilitarian welfare among cardinal allocations weakly decreases after preprocessing, meaning that $\frac{\OPTUSW(I)}{\maxUSW(\mathcal{A})}$ weakly increases. We also remark that the preprocessing does not affect $S,R,$ or $T$.

Before proving the price of cardinality upper bound for the case where $\sum_{i\in R\cup T} u_i(A^*_i) < 1$, we find a lower bound on the utilitarian social welfare of a cardinal allocation, for any arbitrary instance $I$.\footnote{We remark that the lemma holds regardless of whether $\sum_{i\in R\cup T} u_i(A^*_i) < 1$ or $\sum_{i\in R\cup T} u_i(A^*_i) \geq 1$, and whether or not $I$ is preprocessed.}

\begin{lemma}\label{lem:conscard}
    For an arbitrary instance $I$,
    there exists a cardinal allocation $\cA$ such that
    $\USW(\cA)\geq 1+\sum_{i\in S}\frac{k}{|A^*_i|}(u_i(A^*_i)-u_{i^\dagger}(A^*_i))$
    for some agent $i^\dagger\in R$.
\end{lemma}
\begin{proof}[Proof Sketch]
    In the full proof, we show that there exists a cardinal allocation $\cA$ such that
    $\USW(\cA)\geq \sum_{i\in R\cup T}u_i(A^*_i)+\sum_{i\in S}\frac{k}{|A^*_i|}u_i(A^*_i)+\sum_{i\in S}\frac{|A^*_i|-k}{|A^*_i|}u_{i^\dagger}(A^*_i)$,
    which suffices because we have $\sum_{i\in R\cup T}u_i(A^*_i)=1-\sum_{i\in S}u_{i^\dagger}(A^*_i).$ (Recall that $I$ is preprocessed). Specifically, this cardinal allocation can be achieved by the following greedy procedure.
   
    The procedure starts from $\cA^*$, and at each step, reassigns the item with the least utility loss from some \emph{unsatisfied} agent's bundle to some \emph{active} agent;
    an agent is \emph{unsatisfied} if she receives more than $k$ items, and is \emph{active} if she receives less than $k$ items.

\begin{itemize}
    \item Step 1: Set $\cB \leftarrow \cA^*$ as the initial allocation, $P\leftarrow S$ as the initial set of unsatisfied agents, and $Q\leftarrow R$ as the initial set of active agents;
    \item Step 2: If there are no unsatisfied agents, then terminate and output the underlying allocation $\cB$ (this will be $\cA^k$). Otherwise, find the item $e^*\in \bigcup_{i\in P} B_i$ 
    and an active agent $i^*\in Q$ such that reassigning $e^*$ to agent $i^*$ causes the minimum utilitarian social welfare loss among all single-item reassignments from items of unsatisfied agents to active agents. Reassign $e^*$ to agent $i^*$, and update $\cB$ accordingly.
    \item Step 3: Update $P$ and $Q$, and return to Step 2.
\end{itemize}

As $m\leq kn$, the procedure can terminate and the returned allocation $\cA^k$ is cardinal.
Moreover, during the reassignment process, an active agent can never become unsatisfied and any unsatisfied agent can never become active.
\end{proof}

We now prove the upper bound for Theorem \ref{thm:pocusw1-ak} for the case where $\sum_{i\in R\cup T}u_i(A^*_i)< 1$.

\begin{lemma}\label{lem:1catuppercase2}
    If $\sum_{i\in R\cup T}u_i(A^*_i)< 1$, then $$\frac{\OPTUSW(I)}{\maxUSW(\mathcal{A})}\leq \frac{1+s}{1+\frac{ks^2}{m-1}},$$ 
    where $s=-1+\sqrt{1+\frac{m-1}{k}}$.
\end{lemma}
\begin{proof}[Partial Proof]
    We first describe the tiebreak procedure for the utilitarian-optimal allocation $\cA^*$. If multiple agents are tied for having the highest utility for an item, we pick the allocation $\cA^*$ based on the following criteria;
    \begin{itemize}
        \item if it is possible to allocate each item to the agent that values it most, such that all $m$ items are owned by agents with strictly more than $k$ items, then $\cA^*$ is defined as this allocation,
        \item otherwise, the tie is broken in favour of the agent with less than $k$ items.
    \end{itemize}
    
    We divide the remainder of the proof into two cases, depending on whether or not all of the goods are allocated to agents in $S$ under $\cA^*$.

    \textbf{Case 1:} $\sum_{i\in S}|A^*_i|<m$. We present the full proof for the case where not all items are allocated to agents in $S$ under $\cA^*$. 
    Recall that by Lemma~\ref{lem:conscard}, there exists an agent $i^\dagger\in R$ 
    and a cardinal allocation $\cA^k$ such that
    $
    \USW(\cA^k)\geq 1+\sum_{i\in S}\frac{k}{|A^*_i|}(u_i(A^*_i)-u_{i^\dagger}(A^*_i)).
    $

    Consider the agent $i^\dagger \in R$ and a specific item $g^\dagger \in \bigcup_{i\in R\cup T} A^*_i$; the existence of $g^\dagger$ is guaranteed due to $\sum_{i\in S}|A^*_i| < m$.
    We construct another (possibly non-normalized) instance $I'$ which differs from $I$ only by the agents' utilities. Below, we describe the utility function $u'$ that each agent has in $I'$,
    \begin{itemize}
        \item for $i\in R\cup T$, $u'_i(g^\dagger) = 1$ and $u'_i(g)=0$ for all $g\neq g^\dagger$;
        \item for $i\in S$, $u'_i(g) = u_i(g)-u_{i^\dagger}(g)$ if $g\in A^*_i$ and $u'(g)=0$ otherwise.
    \end{itemize}
    Denote by $\cA'$ the utilitarian-optimal allocation of $I'$ and by $S'$ the set of agents receiving more than $k$ items in $\cA'$; note that $S'=S$.
    Due to $\sum_{i\in S}|A^*_i| < m$, when picking $\cA^*$, if multiple agents are tied for having the highest utility for an item, then the tie is broken in favour of the agent with less than $k$ items.
    As a consequence, for any $i\in S$ and $g\in A^*_i$, $u'_i(g) > 0 $ due to $i^\dagger \in R$.

        Now we show that in $I'$, $0<u'_i(A'_i)\leq 1$ for each $i\in S'$. From the construction of $u'_i()$, we immediately have $u'_i(A'_i) \leq 1$.
        To prove $0<u'_i(A'_i)$, since agent $i$ is the only one with positive utility on the items in bundle $A^*_i$, we have $A^*_i \subsetneq A'_i$ and hence $0< u'_i(A'_i)$; note that $u'_i(g)>0$ for every $g\in A^*_i$.

    We now present the upper bound of the ratio regarding $I$ for Case 1 as follows,
    $$
        \begin{aligned}
            \frac{\OPTUSW(I)}{\maxUSW(\mathcal{A})}&\leq \frac{\sum_{i\in R\cup T}u_i(A^*_i)+\sum_{i\in S}u_i(A^*_i)}{\USW(\cA^k)}\\
        &=\frac{1-\sum_{i\in S}u_{i^\dagger}(A^*_i)+\sum_{i\in S}u_i(A^*_i)}{1+\sum_{i\in S}\frac{k}{|A^*_i|}(u_i(A^*_i)-u_{i^\dagger}(A^*_i))}\\
        &=\frac{1+\sum_{i\in S}u'_i(A'_i)}{1+\sum_{i\in S}\frac{k}{|A_i|}u'_i(A'_i)}\\
        &=\frac{\OPTUSW(I')}{\maxUSWa(\mathcal{A})}\\
        & \leq \frac{1}{2}\left( \sqrt{1+\frac{m-1}{k}}+1\right),
        \end{aligned}
    $$
    where the first equality results from the property of the preprocessed instance and the fact that $i^\dagger \in R$;
    the last inequality transition follows from Lemma \ref{lem:fix}.
    \end{proof}

Finally, we conclude the section with the following result on computing utilitarian-optimal cardinal allocations.
\begin{proposition}\label{prop:computil}
    Given a single-category instance $I$ and cardinality constraint $k$, the utilitarian-optimal cardinal allocation can be found in polynomial time, and has a utilitarian social welfare of at least $\frac{2}{1+\sqrt{1+\frac{m-1}{k}}}\cdot\OPTUSW(I)$. 
\end{proposition}
\begin{proof}
  Consider a complete bipartite graph $G=(U,V,E)$, where $U$ represents $k$ copies of each agent, and $V$ represents the $m$ goods, with zero-valued dummy items added such that $|U|=|V|$. Also, an edge between agent $i\in U$ and item $g\in V$ has weight equal to $u_i(g)$. Our desired allocation can be found by computing a maximum weight bipartite matching, such as by using the Hungarian algorithm \cite{https://doi.org/10.1002/nav.3800020109}. The utilitarian social welfare guarantee follows immediately from Theorem~\ref{thm:pocusw1-ak}.
\end{proof}
\subsection{Egalitarian Social Welfare}
We now move to the objective of egalitarian social welfare, where the worst-case degradation of worst-case fairness objective is quantified by our exact and tight bounds on the egalitarian price of cardinality. Note that in addition to the assumption that $m>k$, we also assume in this subsection that $m\geq n$, because if $m<n$, then $\OPTESW(I)=0$ and consequently, the egalitarian price of cardinality will be $1$.
\begin{theorem}\label{thm::ega-poc}
    In the single-category case, the egalitarian price of cardinality is $\max\left\{\frac{m-n+1}{k},1\right\}$. 
\end{theorem}
\begin{proof}
    We begin by proving the upper bound and consider instance $I$. If $\OPTESW(I)=0$, the statement holds trivially, so we assume $\OPTESW(I)> 0$ for the remainder of the proof. 
    Denote by $\cA^*=(A^*_1,\ldots,A^*_k)$ the unconstrained egalitarian-optimal allocation. If $|A^*_i| \leq k$ holds for each $i\in N$, then the egalitarian-optimal cardinal allocation is the same as $\cA^*$, and thus the theorem statement holds for this case. We now focus on the remaining case where $\max_{i\in N} |A_i| > k$.

    Let $S$ be the set of agents who each receive more than $k$ items in $\cA^*$.
    We now consider a cardinal allocation $\cA'$ where each $i\in S$ receives their $k$ most valued items from $A^*_i$ and for each agent $j\notin S$, $A^*_j\subseteq A'_j$; note that $\cA'$ can be cardinal due to $m\leq kn$.
    Then we have the following,
    $$
    \begin{aligned}
    \ESW(\cA') &= \min \left\{ \min_{i\in S} u_i(A'_i), \min_{i\notin S} u_i(A'_i)  \right\}\\
    &\geq \min \left\{ \min_{i\in S} k \cdot \frac{u_i(A^*_i)}{|A^*_i|}, \OPTESW(I) \right\} \\
    & \geq k\cdot \frac{\OPTESW(I)}{m-n+1},
    \end{aligned}
    $$
    where the last inequality transition is due to $|A^*_j| \geq 1$ for all $j\in N$ (and hence $|A^*_i| \leq m-n+1$ for all $i$).
    As a consequence, we have 
    $$
    \frac{\OPTESW(I)}{\maxESW(\mathcal{A})} \leq \frac{\OPTESW(I)}{\ESW(\cA')}  \leq \frac{m-n+1}{k},
    $$
    completing the proof of the upper bound.
    
    For the lower bound, consider instance $I$ where the agents' utility functions are as follows;
    \begin{itemize}
        \item for $i\in[n-1]$, $u_i(g_i) = 1$ and $u_i(g_j)=0$ for all $j\neq i$;
        \item for $i=n$, $u_n(g_j)=0$ for each $j\in [n-1]$ and $u_n(g_j) = \frac{1}{m-n+1}$ for each $j \geq n$.
    \end{itemize}
    By allocating each $g_i$ to $i$ for each $i\in [n-1]$ and allocating all of the remaining items to agent $n$, we see that $\OPTESW(I)=1$. On the other hand, $\maxESW(\mathcal{A})=\frac{k}{m-n+1}$ as agent $k$ can receive at most $k$ items, deriving the desired lower bound.
\end{proof}

Note that this bound is tight for all feasible values of $m$, $n$, and $k$. This result shows that when $m$ is large compared to $n$ and $k$, there may be a significant reduction in egalitarian fairness when cardinality constraints are naively imposed in pursuit of a fair allocation. We also remark that although computing an egalitarian-optimal (possibly non-cardinal) allocation is well-known to be NP-hard \cite{DBLP:conf/coco/Karp72}, if we are provided with such an allocation, we can find, in linear time, a cardinal allocation with an egalitarian social welfare guarantee corresponding to the egalitarian price of cardinality. This is simply achieved by letting each agent keep their $k$ most valued items from the starting egalitarian-optimal allocation.

\section{Multiple Categories}\label{sec::multi-cate}
We now extend our analysis to the setting where the items are partitioned into multiple categories. Recall that there are $h$ categories, where category $j\in [h]$ has $m_j$ items to be allocated and a cardinality constraint of $k_j$. 
We also ensure that all items can be assigned, $\frac{m_j}{k_j} \leq n$ holds for each $j\in [h]$.
Without loss of generality, we order categories such that $\frac{k_1}{m_1}\leq \frac{k_2}{m_2}\leq \dots \leq \frac{k_h}{m_h}$ and break ties in favour of the category with a smaller number of items (i.e., if $\frac{k_i}{m_1}=\frac{k_j}{m_j}$ and $m_i<m_j$, then set $i<j$).

\subsection{Utilitarian Social Welfare}

For utilitarian social welfare, we first consider the case of two agents. Before stating the main result, we establish a key reduction which restricts the space of instances to those with weakly higher $\frac{\OPTUSW(I)}{\maxUSW(\mathcal{A})}$ ratio.

\begin{lemma}\label{lem:red}
    Given an instance $I$ with two agents and cardinality constraints $\kappa$, there exists another instance $I'$ which only differs from $I$ in the utility functions, where:
    \begin{itemize}
        \item under the utilitarian-optimal allocation $\cA'^*$,
        \begin{itemize}
            \item both agents exceed the cardinality constraint in exactly one category each,
            \item neither agent receives any utility from any category where they do not exceed the cardinality constraint,
        \end{itemize}
        \item and $\frac{\OPTUSW(I)}{\maxUSW(\mathcal{A})}\leq \frac{\OPTUSW(I')}{\maxUSWa(\mathcal{A})}$ holds.
    \end{itemize}
\end{lemma}

Following this reduction, we are now ready to present the utilitarian price of cardinality for two agents, which is exact and tight for all $\kappa=(k_1,\dots,k_h)$ and $m_1,\dots,m_h$.
\begin{theorem}\label{thm:multi-n=2}
    For two agents and $h\geq 2$,  
    the utilitarian price of cardinality is $\frac{2}{\frac{k_1}{m_1} + \frac{k_2}{m_2}}$.
\end{theorem}
\begin{proof}
    By Lemma~\ref{lem:red}, it suffices to focus on the case where in $\cA^*$, agent $1$ (resp. agent $2$) exceeds the cardinality constraint of category $j_1$ (resp. $j_2$). Note that due to the ordering of our categories, it is `weakly better' to consider categories $1$ and $2$.
    Moreover in $\cA^*$, each agent $i$ only receives non-zero utility from $C_{j_i}$. Note that we must have $j_1\neq j_2$ due to $\frac{m_j}{k_j} \leq 2$ for all $j$.

    We first prove the upper bound. Consider another (possibly non-normalized) instance $I'$ that only differs from $I$ in utility functions. In $I'$, agent $i$ has utility function $u'_i(g)=u_i(g)$ if $g\in A^*_{ij_i}$ and $u'_i(g)=0$ otherwise.
    One can verify that the welfare of utilitarian-optimal allocation of $I$ is equal to that of $I'$, while the maximum welfare of cardinal allocations is weakly decreased in $I'$. Accordingly, we have $\frac{\OPTUSW(I)}{\maxUSW(\mathcal{A})}\leq \frac{\OPTUSW(I')}{\maxUSWa(\mathcal{A})}$.
    We then convert $I'$ into a normalized instance $I''$ by increasing agent $i$'s utility for $A^*_{ij_i}$ to $1$ in a way such that the `price of cardinality' ratio weakly increases; this can be done by increasing the utility of items valued the most by both agents. 
    Note that $I''$ is a normalized instance where in $\cA^*$, each agent $i$ receives utility $1$ from obtaining all of the items which they positively value.

    Finally, we have
    \begin{align*}  
    \frac{\OPTUSW(I)}{\maxUSW(\mathcal{A})}&\leq \frac{\OPTUSW(I'')}{\maxUSWb(\mathcal{A})}\\
    &\leq \frac{2}{\frac{k_{j_1}}{|A^*_{1j_1}|} + \frac{k_{j_2}}{|A^*_{2j_2}|}}\\
    &\leq \frac{2}{\frac{k_{j_1}}{m_{j_1}} + \frac{k_{j_2}}{m_{j_2}}}\leq \frac{2}{\frac{k_1}{m_1} + \frac{k_2}{m_2}},
    \end{align*}
    concluding the proof of the upper bound.
    
    For the lower bound, consider the instance $I$ where agent $1$ values each item in category $1$ at $\frac{1}{m_1}$ utility and agent $2$ values each item in category $2$ each at $\frac{1}{m_2}$ utility. Clearly, $\frac{\OPTUSW(I)}{\maxUSW(\mathcal{A})}=\frac{2}{\frac{k_1}{m_1} + \frac{k_2}{m_2}}$ for this instance.
\end{proof}
We give the utilitarian price of cardinality for general $n$.
\begin{theorem}\label{thm:multi-general}
    For general $n$, the utilitarian price of cardinality is $\frac{m_1}{k_1}$.
\end{theorem}
\begin{proof}
    We first prove the upper bound. Given an instance $I$, let $\cA^*$ be its utilitarian-optimal allocation. Then we have 
    $$
    \begin{aligned}
    &\frac{\OPTUSW(I)}{\maxUSW(\mathcal{A})}\\
    & \leq \frac{ \sum_{j\in [h]}\sum_{i\in S_j}u_i(A^*_{ij}) + \sum_{j\in [h]}\sum_{i\in N\setminus S_j} u_i(A^*_{ij}) }{\sum_{j\in [h]}\sum_{i\in S_j} \frac{k_j}{|A^*_{ij}|} u_i(A^*_{ij}) + \sum_{j\in[h]}\sum_{i\in N\setminus S_j} u_i(A^*_{ij})} \\
    & \leq \frac{ \sum_{j\in [h]}\sum_{i\in S_j} u_i(A^*_{ij})}{\sum_{j\in [h]}\sum_{i\in S_j} \frac{k_j}{|A^*_{ij}|} u_i(A^*_{ij})} \leq \frac{m_1}{k_1},
    \end{aligned}
    $$
    where the last inequality transition is because for every $j\in [h]$, $\frac{k_j}{|A^*_{ij}|} \leq \frac{k_j}{m_j} \leq \frac{k_1}{m_1}$.

    For the lower bound, consider an instance $I$ with $h=n$ categories, and where each category $j\in [h]$ has the same cardinality constraint of $k_j=k$ and same number of items $m_j=q$ items, where $q$ is divisible by $k$.
    Suppose that in this instance, for each agent $i$, $u_i(g) = \frac{1}{q}$ if $g\in C_i$, and $u_i(g)=0$ otherwise.
    Clearly, we have $\OPTUSW(I) = n$.
    In the utilitarian-optimal cardinal allocation $\cA^k$, each agent $i$ receives a utility of $\frac{k}{q}$. Thus, we have $\USW(\cA^k) = \frac{nk}{q}$, and therefore, the utilitarian price of cardinality is at least 
    $
    \frac{\OPTUSW(I)}{\USW(\cA^k)} = \frac{q}{k} = \frac{m_1}{k_1},
    $
    completing the proof.
\end{proof}
This result is roughly tight in the sense that the utilitarian price of cardinality is at most $\frac{m_1}{k_1}$ for any instance with cardinality constraints $\kappa=(k_1,\dots,k_j)$, and is precisely $\frac{m_1}{k_1}$ for any instance where there are at least $n$ categories, and $\frac{k_1}{m_1}=\frac{k_2}{m_2}=\dots=\frac{k_n}{m_n}$. 

Finally, we mention a result on computing utilitarian-optimal cardinal allocations, similar to Proposition~\ref{prop:computil}.
\begin{proposition}\label{prop:computilmulti}
    Given a multiple-category instance $I$ and cardinality constraints $\kappa$, the utilitarian-optimal cardinal allocation can be found in polynomial time, and has a utilitarian social welfare of at least $\frac{k_1}{m_1}\cdot\OPTUSW(I)$. 
\end{proposition}
\begin{proof}
  The proof is almost identical to the proof of Proposition~\ref{prop:computil}, but we instead construct a separate complete bipartite graph for each category $j\in [h]$, and compute a maximum weight bipartite matching for each of these graphs. The runtime remains in polynomial time, and the utilitarian social welfare guarantee follows immediately from Theorem~\ref{thm:pocusw1-ak}.
\end{proof}
\subsection{Egalitarian Social Welfare}
Finally, for egalitarian social welfare, we present bounds for the price of cardinality which are exact and tight for any $n$, $\kappa=(k_1,\dots,k_h)$, and $m_1,\dots,m_h$.
\begin{theorem}\label{thm:multi-egal}
    If $n\leq \sum_{j=2}^hm_j+1$, then the egalitarian price of cardinality is $\frac{m_1}{k_1}$. If $n>\sum_{j=2}^hm_j+1$, then the egalitarian price of cardinality is
    $$\max_{j\in [h]}\left\{\frac{m_j-\max\{n-1-\sum_{t\neq j}m_t,0\}}{k_j}\right\}.$$
\end{theorem}
\begin{proof}
    We begin with the upper bound and fix instance $I$. Let $\cA^*=(A^*_1,\ldots,A^*_n)$ be the egalitarian-optimal allocation for $I$ and let $A^*_{ij}$ be the bundle agent $i$ receives from category $j$. If $\ESW(\cA^*)=0$, then the statement trivially holds. We further assume $\ESW(\cA^*)>0$.

    We first show a general upper bound of $\frac{m_1}{k_1}$. Note that in the egalitarian-optimal cardinal allocation, every agent $i$ can keep their $k_j$ most valued items and receives a utility of at least $\sum_{j\in [h]}\frac{k_j}{m_j}u_i(A^*_{ij})\geq \frac{k_1}{m_1}u_i(A^*_{i}) \geq \frac{k_1}{m_1}\ESW(\cA^*)$.
    This bound holds for any $n$ and $m_j$, proving the first part of the theorem statement.

    We then strengthen the bound for the case where $n>\sum_{j=2}^hm_j+1$. Define, for each $j\in[h]$, $c_j:=\max\{ n-1-\sum_{t\neq j}m_t, 0\}$.
    By the definition of $c_j$, we claim that for any $i\in [n]$ and $j\in [h]$, $|A^*_{ij}| \leq m_j- c_j$ holds; otherwise, there must be one agent receiving no item in $\cA^*$, contradicting $\ESW(\cA^*)>0$.
    Then in the egalitarian-optimal cardinal allocation, each agent $i$ can obtain a utility of at least
    \begin{align*}
       \sum_{j\in[h]} \frac{k_j}{|A^*_{ij}|}u_i(A^*_{ij}) &\geq \sum_{j\in[h]} \frac{k_j}{m_j-c_j}u_i(A^*_{ij})\\
       &\geq \min\limits_{j\in[h]} \frac{k_j}{m_j-c_j}\sum_{j\in[h]}u_i(A^*_{ij})\\
       &= \min\limits_{j\in[h]} \frac{k_j}{m_j-c_j} u_i(A^*_i). 
    \end{align*}
    
    Therefore, the price of cardinality in this case is at most $\max_{j\in [h]}\frac{m_j-c_j}{k_j}$, equal to the expression in the theorem statement.

    For the lower bound, we first consider the case where $n\leq \sum_{j=2}^hm_j+1$ and show that the bound of $\frac{m_1}{k_1}$ is tight. Let agent 1 value each item of $C_1$ at $\frac{1}{m_1}$ utility each and let agents $2,\ldots,n$ only value one unique item each at 1 utility from $C_2,\ldots,C_h$.
    In the egalitarian-optimal allocation, each agent receives a utility of 1, leading to an optimal egalitarian welfare of 1. However, for cardinal allocations, the utility of agent 1 is at most $\frac{k_1}{m_1}$ and the lower bound of $\frac{m_1}{k_1}$ follows.

    For the case where $n>\sum_{j=2}^hm_j+1$, we first denote the category that maximizes $\max_{j\in[h]} \frac{m_j-c_j}{k_j}$ as $j^*$; recall that $c_j=\max\{ n-1-\sum_{t\neq j}m_t, 0 \}$.
    In the subcase where $c_{j^*}=0$, we let agents $1,\ldots,n-1$ value one unique item each at $1$ utility in any category $j\neq j^*$ and agent $n$ values each item in $C_{j^*}$ evenly, achieving at most $\frac{k_{j^*}}{m_{j^*}}$ utility in a cardinal allocation. Since the optimal egalitarian welfare is $1$, the respective lower bound follows.
    As for the subcase where $c_{j^*}>0$ (this indeed implies $j^*=1$), consider the instance $I$ where $\sum_{j\neq j^*}m_j$ of the agents value one unique item each at $1$ utility in some $C_j$ with $j\neq j^*$. Here, $n-1-\sum_{j\neq j^*} m_j$ of the remaining agents value one unique item each at $1$ utility in $C_{j^*}$, and the last agent values all $m_{j^*}-c_{j^*}$ remaining items in $C_{j^*}$ evenly.
    We have $\OPTESW(I)=1$, and note that the last agent can receive a utility of at most $\frac{k_{j^*}}{m_{j^*}-c_{j^*}}$ in any cardinal allocation, achieving the lower bound in the statement.
\end{proof}

Similar to the single category scenario, if we are given an egalitarian-optimal allocation, we can construct a cardinal allocation with an egalitarian welfare guarantee corresponding to the price of cardinality by letting each agent keep their $k_j$ most valued items in each category $j\in [h]$.
\section{Discussion}
In this work, we introduced the utilitarian and egalitarian prices of cardinality, which quantify the worst-case multiplicative loss of social welfare when cardinality constraints are imposed on the allocation. For both the single- and multi-category cases, we present tight bounds on the prices of cardinality, expressed as an exact (rather than asymptotic) function of the instance and cardinality parameters. Our results enable decision makers to make a clear, well-informed choice of cardinality constraint with respect to the level of balancedness and the potential loss of social welfare.

Our parametrized approach to the price of cardinality can be applied to other parametrized notions of fairness such as envy-freeness up to $k$ items (EF-$k$) or $\alpha$-maximin share ($\alpha$-MMS), providing similar insights to decision makers regarding the tradeoff between the level of fairness and the potential loss of social welfare.

An immediate open question is to find a more precise utilitarian price of cardinality in the case of multiple categories and $n\geq 3$ agents. Ideally, we would like the price to be tight for all possible values of $n$, $\{m_j\}_{j\in [h]}$, and $(k_j)_{j\in [h]}$, like our egalitarian prices of cardinality.
To pursue such a precise price, one approach is to characterize the worst case scenario. Unfortunately, our reduction (in Lemma \ref{lem:red}) for the case of two agents can not be immediately extended to the case where $n\geq 3$. Furthermore, part of the complete proof of Lemma~\ref{lem:1catuppercase2} relies on knowing the exact form of a parameter $s$, which we cannot find in the multi-category case due to the multivariate property of the problem. However, we believe that the greedy procedure in the proof of Lemma \ref{lem:conscard} could be helpful for identifying the worst case structure.

Another possible direction is finding the price of cardinality for the ``dual'' problem where each agent must receive \emph{at least} $k$ items. However, this type of constraint violates the hereditary property of matroids, and is generally not studied in other related work.

\section*{Acknowledgements}
The authors would like to thank Siddarth Barman, Xiaohui Bei, and the anonymous AAAI reviewers for their helpful comments. This work is funded by HKSAR RGC under Grant No. PolyU 15224823 and the Guangdong Basic and Applied Basic Research Foundation under Grant No. 2024A1515011524.

\bibliographystyle{alpha}
\bibliography{aaai25.bib}

\newpage
\appendix
\section{Omitted Proofs for Section~\ref{sec::single-cat}}
\subsection{Proof of Lemma~\ref{lem:1catuswlowergen}}
By applying basic calculus (which will also be shown in the proof of Lemma~\ref{lem:fix}), we find that $s=-1+\sqrt{1+\frac{m-1}{k}}$ is the unique global maximum of the price of cardinality lower bound of $\frac{1+s}{1+\frac{ks^2}{m-1}}$ (under the domain of $s>0$). As a result, if $m$ and $k$ are such that $s$ is not integral, we have another instance which shows that the utilitarian price of cardinality is at least $\frac{1}{2}(1+\sqrt{1+\frac{m-1}{k}})-1$. 

We now fix $m,n,k \in \mathbb{N}^{+}$ satisfying $n \geq \frac{m}{k}$ and $m-2\geq k$; note that if $m-1=k$, the utilitarian-optimal allocation is cardinal as the utilities are normalized. Let $s=-1+\sqrt{1+\frac{m-1}{k}}$ and denote $t\coloneqq \max\{ \lfloor s \rfloor, 1\}$. Note that for any $n\geq 2$, we have $t \leq n-1$.
Now consider an instance $I$ with $n$ agents and $m$ items, where the agents' utilities are described as follows;
\begin{itemize}
    \item for $i=1,\ldots,t$, if $j=(i-1)\lfloor \frac{m-1}{t} \rfloor + 1,\ldots, i\lfloor \frac{m-1}{t} \rfloor$, $u_i(g_j) = \frac{1}{\lfloor \frac{m-1}{t} \rfloor}$; otherwise, $u_i(g)=0$;
    \item for $i\geq t+1$, $u_i(g_m)=1$ and $u_i(g)=0$ for every $g\neq g_m$.
\end{itemize}
Note that for agent $t$, each item she values with non-zero utility has an index of $t\lfloor \frac{m-1}{t} \rfloor < m$ and thus the above instance is well-defined.
One can verify that $\OPTUSW(I)=1+t$.
Moreover, in the utilitarian-optimal allocation, we know that every agent $i\leq t$ exceeds the cardinality constraint as:
\begin{itemize}
    \item if $t=\lfloor s \rfloor$, we have $\frac{m-1}{t}>\frac{m-1}{s}>k$, and since $k\in \mathbb{N}^+$, we have $\lfloor\frac{m-1}{t}\rfloor>k$,
    \item if $t=1$, agent 1 receives $m-1>k$ items.
\end{itemize}

Accordingly, we have $\maxUSW(\mathcal{A}) = 1+\frac{kt}{\lfloor \frac{m-1}{t} \rfloor}$, showing that the utilitarian price of cardinality is at least 
$
\frac{1+t}{1+\frac{kt}{\lfloor \frac{m-1}{t} \rfloor}}$. Recall that $\frac{1}{2}(1+\sqrt{1+\frac{m-1}{k}})$ is equal to $\frac{1+s}{1+\frac{ks^2}{m-1}}$, and therefore to prove the theorem statement, it suffices to show that
$$
\frac{1+s}{1+\frac{ks^2}{m-1}} - \frac{1+t}{1+\frac{kt}{\lfloor \frac{m-1}{t} \rfloor}} \leq 1,
$$
or equivalently,
$$
s-t+\frac{kts}{\lfloor \frac{m-1}{t} \rfloor} - \frac{kts^2}{m-1} \leq 1+\frac{2ks^2}{m-1}+\frac{k^2s^2t}{(m-1)\lfloor \frac{m-1}{t}\rfloor}.
$$
Since $s-t\leq 1$, it further suffices to show that 
$$
\frac{kts}{\lfloor \frac{m-1}{t} \rfloor} - \frac{kts^2}{m-1} \leq \frac{2ks^2}{m-1}+\frac{k^2s^2t}{(m-1)\lfloor \frac{m-1}{t}\rfloor},
$$
or equivalently, 
$$
\left(\frac{2s}{t}+s\right)\lfloor \frac{m-1}{t} \rfloor + ks \geq m-1.
$$
We have
$$
\begin{aligned}
\left(\frac{2s}{t}+s\right)\lfloor \frac{m-1}{t} \rfloor + ks&\geq (2+s)\lfloor \frac{m-1}{t} \rfloor + ks\\
& > s\left(\frac{m-1}{t}-1\right)  +ks \\
& \geq \frac{s}{t}m-1+(k-1)s \\
& \geq m-1,
\end{aligned}
$$
proving that the utilitarian price of cardinality is at least $\frac{1}{2}(1+\sqrt{1+\frac{m-1}{k}}) -1$.

\subsection{Proof of Lemma~\ref{lem:fix}}
    Firstly, we know that $$\OPTUSW(I)=\sum_{i\in N\setminus S}u_i(A^*_i)+\sum_{i\in S}u_i(A^*_i)$$ and $$\maxUSW(\mathcal{A})\geq \sum_{i\in N\setminus S}u_i(A^*_i)+k\sum_{i\in S}\frac{u_i(A^*_i)}{|A^*_i|}.$$ Note that letting each agent $i\in S$ discard their $|A^*_i|-k$ least valued items yields in a (cardinal) partial allocation with the desired utilitarian welfare guarantee.
    We therefore have
    \begin{align*}
        \frac{\OPTUSW(I)}{\maxUSW(\mathcal{A})}&\leq \frac{\sum_{i\in N\setminus S}u_i(A^*_i)+\sum_{i\in S}u_i(A^*_i)}{\sum_{i\in N\setminus S}u_i(A^*_i)+k\sum_{i\in S}\frac{u_i(A^*_i)}{|A^*_i|}}\\
        &\leq \frac{1+\sum_{i\in S}u_i(A^*_i)}{1+k\sum_{i\in S}\frac{u_i(A^*_i)}{|A^*_i|}},
    \end{align*}
    where the last inequality is due to $\frac{a+b}{a+c}\leq \frac{b}{c}$ for $a,b,c>0$. This proves the first inequality of the lemma statement.
    
    We next show the second inequality through derivatives.
    Consider the multivariable function $F(u_1(A^*_1),\ldots,u_{|S|}A^*_{|S|})=\frac{1+\sum_{i\in S}u_i(A^*_i)}{1+k\sum_{i\in S}\frac{u_i(A^*_i)}{|A^*_i|}}$ with domain $0\leq u_i(A^*_i) \leq 1$ for all $i\in S$. 
    For each $i' \in S$, the derivative $\frac{\partial F}{\partial u_{i'}(A^*_{i'})}$ of the function $F$ with respect to $u_{i'}(A^*_{i'})$ is
    $$
    \frac{\frac{|A^*_{i'}|-k}{|A^*_{i'}|}+k\left(\sum_{i\in S\setminus \{i'\}}u_i(A^*_i)\left(\frac{1}{|A^*_i|}-\frac{1}{|A^*_{i'}|}\right)\right)}{\left(1+k\sum_{i\in S}\frac{u_i(A^*_i)}{|A^*_i|}\right)^2}.
    $$
    The numerator of the derivative is independent of $u_{i'}(A^*_{i'})$, so the expression has no stationary points and is either monotonic increasing, decreasing or constant. 
    Therefore, $F$ is maximized when every $u_i(A^*_i)$ is either 1 or 0; note that for any $i\in S$, $u_i(A^*_i) \in [0,1]$. Also note that for a specific $i'\in S$, setting $u_{i'}(A^*_{i'})=0$ may violate the property of $i' \in S$. However, as we are only concerned with establishing an upper bound for $F$, letting $u_{i'}(A^*_{i'})=0$ for some $i'$ does not result in any violation.

    Denote by $S'\subseteq S$ the set of agents who have utility $1$ in the solution maximizing $F$. Observe that for there exists $i\in S$ such that the derivative of $F$ with respect to some $u_i(A^*_i)$ is positive (e.g. the agent whose bundle under $\cA^*$ has the most items). Therefore $S' \neq \emptyset$, so the following holds:
    $$
    F=\frac{1+\sum_{i\in S}u_i(A^*_i)}{1+k\sum_{i\in S}\frac{u_i(A^*_i)}{|A^*_i|}} \leq \frac{1+ \sum_{i \in S'} 1}{1+k\sum_{i \in S'}\frac{1}{A^*_i}} \leq \frac{1+|S'|}{1+\frac{k|S'|^2}{m-1}},
    $$
    where the first inequality is due to the construction of $S'$. The second inequality is due to the arithmetic mean-harmonic mean (AM-HM) inequality and the fact that $\sum_{i\in S}|A^*_i| \leq m-1$ (the agents of $N\setminus S$ must receive at least one item in $\cA^*$ so that the sum of their utilities can be 1).
    Now consider the function $G(|S'|) = \frac{1+|S'|}{1+\frac{k|S'|^2}{m-1}}$ with domain $0\leq |S'| \leq n$. Its derivative with respect to $|S|$ is
    $$
    \frac{\partial G}{\partial |S'|} = \frac{1}{\left( 1+\frac{k|S'|^2}{m-1} \right)^2}\cdot \left( 1-\frac{k|S'|(2+|S'|)}{m-1} \right).
    $$
    This derivative is non-negative for $-1-\sqrt{1+\frac{m-1}{k}}\leq |S'| \leq -1+\sqrt{1+\frac{m-1}{k}}$. Since $m\leq kn$ always holds, we have $-1+\sqrt{1+\frac{m-1}{k}} \leq \sqrt{n} < n$ for all $n\geq 2$.
    Hence, $G$ achieves its maximum value when $|S'| = -1+\sqrt{1+\frac{m-1}{k}}$.
    Therefore, by setting $s=-1+\sqrt{1+\frac{m-1}{k}}$, we have 
    \begin{align*}
    \frac{\OPTUSW(I)}{\maxUSW(\mathcal{A})} &\leq \max_{0\leq u_i(A^*_i)\leq 1}F\\
    &\leq \max_{0\leq |S'| \leq n}G\\
    &\leq  \frac{1+s}{1+\frac{ks^2}{m-1}},
    \end{align*}
    completing the proof.

\subsection{Proof of Lemma~\ref{lem:conscard}}
    Specifically, we will show that there exists a cardinal allocation $\cA$ such that
    \begin{align*}
    \USW(\cA)&\geq \sum_{i\in R\cup T}u_i(A^*_i)+\sum_{i\in S}\frac{k}{|A^*_i|}u_i(A^*_i)\\
    &\quad +\sum_{i\in S}\frac{|A^*_i|-k}{|A^*_i|}u_{i^\dagger}(A^*_i),
    \end{align*}
    which suffices because $$\sum_{i\in R\cup T}u_i(A^*_i)=1-\sum_{i\in S}u_{i^\dagger}(A^*_i).$$ (Recall that $I$ is preprocessed.)

    Fix an agent $t\in R$. For each $i\in S$, denote $L_i\subseteq A^*_i$ as the set of $|A^*_i|-k$ items which minimizes the loss of utilitarian social welfare when they are reallocated from $i$ to agent $t$. Accordingly, for each $i\in S$, we have
    $$
    \frac{\sum_{g\in L_i} ( u_i(g)-u_t(g))}{|A^*_i|-k}\leq \frac{u_i(A^*_i)-u_t(A^*_i)}{|A^*_i|},$$
    which is equivalent to
    $$ \sum_{g\in L_i} (u_i(g)-u_t(g)) \leq \frac{|A^*_i|-k}{|A^*_i|}(u_i(A^*_i)-u_t(A^*_i)).
    $$
    Hence, the total loss of utilitarian social welfare from reallocating $L_i$ from every $i\in S$ to agent $t$ is at most $\sum_{i\in S}\frac{|A^*_i|-k}{|A^*_i|}(u_i(A^*_i)-u_t(A^*_i))$, and therefore,
    the utilitarian social welfare of the resulting allocation (which may not be cardinal) is at least 
    $$\sum_{i\in R\cup T}u_i(A^*_i)+\sum_{i\in S}\frac{k}{|A^*_i|}u_i(A^*_i)+\sum_{i\in S}\frac{|A^*_i|-k}{|A^*_i|}u_t(A^*_i).$$

    As letting one agent $t\in R$ receive all excess items from $\{A^*_i\}_{i\in S}$ may not yield a cardinal allocation, we then consider the cardinal allocation achieved by reassigning items from $\{A^*_i\}_{i\in S}$ to agents $R$ via a greedy procedure described as follows.
    
    The procedure starts from $\cA^*$, and at each step, reassigns the item with the least utility loss from some \emph{unsatisfied} agent's bundle to some \emph{active} agent;
    an agent is \emph{unsatisfied} if she receives more than $k$ items, and is \emph{active} if she receives less than $k$ items.
    We consider the following reassignment process:
\begin{itemize}
    \item Step 1: Set $\cB \leftarrow \cA^*$ as the initial allocation, $P\leftarrow S$ as the initial set of unsatisfied agents, and $Q\leftarrow R$ as the initial set of active agents;
    \item Step 2: If there are no unsatisfied agents, then terminate and output the underlying allocation $\cB$ (this will be $\cA^k$). Otherwise, find the item $e^*\in \bigcup_{i\in P} B_i$ 
    and an active agent $i^*\in Q$ such that reassigning $e^*$ to agent $i^*$ causes the minimum utilitarian social welfare loss among all single-item reassignments from items of unsatisfied agents to active agents. Reassign $e^*$ to agent $i^*$, and update $\cB$ accordingly.
    \item Step 3: Update $P$ and $Q$, and return to Step 2.
\end{itemize}

As $m\leq kn$, the procedure can terminate and the returned allocation $\cA^k$ is cardinal.
Moreover, during the reassignment process, an active agent can never become unsatisfied and any unsatisfied agent can never become active.

Without loss of generality, let agent $n$ be the chosen active agent at the last reassignment step. Note that $n\in R$ must hold.
For any $j\in S$, we define the loss of utility from reassigning item $e\in A^*_j$ to agent $n$ as $l^n_j(e) \coloneqq u_{j}(e) - u_n(e)$, and moreover, let $L^n_j$ be the set of the $|A^*_j| - k$ items of $A^*_j$ with the lowest $l^n_j(e)$.
By this construction, $\sum_{ j\in S} \sum_{g \in L^n_j} l_j^n(g)$ is the welfare loss of reassigning $\bigcup_{j \in S}\bigcup_{ g\in L^n_j} g$ to agent $n$.
Based on the above arguments and the fact that $n\in R$, we have $\sum_{ j\in S} \sum_{g \in L^n_j} l_j^n(g) \leq \sum_{i\in S}\frac{|A^*_i|-k}{|A^*_i|}(u_i(A^*_i)-u_{n}(A^*_i))$.
Our last step is to show that $\USW(\cA^*)-\USW(\cA^k)\leq \sum_{ j\in S} \sum_{g \in L^n_j} l_j^n(g)$.

Consider an arbitrary agent $j\in S$. According to the reassignment process, $|A^*_j| - k$ items of $A^*_j$ are not allocated to $j$ in $\cA^k$. 
Suppose $0< \delta_1 \leq \delta_2\leq \cdots \leq \delta_{|A^*_j|-k}$, where $\delta_i$ refers to the $i$-th lowest utility loss in the aforementioned reassignment of $|A^*_j| - k$ items. 
Recall that $L^n_j$ is the set of the $|A^*_j| - k$ items of $A^*_j$ with the least $l_j^n(e)$. Let $L^n_j = \{ g_1,g_2,\ldots,g_{|A^*_j| - k} \}$ and $l^n(g_t) \leq l^n(g_{t+1}) $ for all $t \leq |A^*_j| - k$.
\begin{claim}\label{claim:ak-1}
    For any $1 \leq t \leq |A^*_j| - k$, it holds that $\delta_t \leq l_j^n(g_{t})$.
\end{claim}
\begin{proof}[Proof of Claim~\ref{claim:ak-1}]
    For a contradiction, let $t'$ be the smallest index such that $\delta_{t'} > l^n(g_{t'})$.
    Since agent $n$ is the chosen active agent at the last reassignment step, they were active when the reassignment corresponding to $\delta_{t'}$ happened.
    For ease of presentation, we let round $s$ be the moment when the reassignment corresponding to welfare loss $\delta_{t'}$ happened, and accordingly, agent $j$ is unsatisfied at the start of round $s$.
    By $\delta_1 \leq \delta_2 \leq \cdots \leq \delta_{t'}$ and the reassignment rule, for any $t < t'$, the reassignment corresponding to welfare loss $\delta_t$ happened before round $s$.
    There are two possible cases depending on whether item $g_{t'}\in L_j^n$ is still in agent $j$'s bundle at the start of round $s$. If so, then reassigning $g_{t'}$ to agent $n$ results in a welfare loss $l^n(g_{t'}) < \delta_{t'}$, contradicting Step 2.
    
    We now consider the remaining case where item $g_{t'}$ is \textbf{not} in the bundle of agent $j$ at the start of round $s$. Since $g_{t'} \in A^*_j$, then the reassignment of $g_{t'}$ in the process results in a welfare loss of $\delta_{t''}$ where $t'' < t'$, and accordingly, the reassignment of $g_{t'}$ happened before round $s$.
    Note that at the end of round $s-1$, at most $t'-1$ items were reassigned from bundle $A^*_j$. Since $g_{t'}$ was reassigned before round $s$, then 
    at least one item from $\{ g_1, g_2,\ldots, g_{t'-1} \}$ is in the bundle of agent $j$ at the start of round $s$. Reassigning any item from $\{ g_1, g_2,\ldots, g_{t'-1} \}$ to agent $n$ induces a utility loss at most $l^n(g_{t'})$; note that $l^n(g_r) \leq l^n(g_{r+1}) $ for all $r \leq |A^*_j| - k$.
    Thus, at the start of round $s$, agent $j$ is unsatisfied and their bundle contains some item from $\{ g_1, g_2,\ldots, g_{t'-1} \}$.
    Therefore, at the start of round $s$, there exists some item $g$ in the bundle of some agent $j$ (an unsatisfied agent), and some agent $n$ (an active agent), such that reassigning $g$ to agent $n$ causes welfare loss at most $l^n(g_{t'}) < \delta_{t'}$, a contradiction.   
\end{proof}
The above claim implies that in the greedy reassignment procedure, for any $j\in S$, the welfare loss of reassigning items in 
$A^*_j$ is at most $\sum_{g \in L^n_j} l_j^n(g)$. Then, by summing up over all $j\in S$, we have
\begin{align*}
\USW(\cA^*) - \USW(\cA^k) &\leq \sum_{ j\in S} \sum_{g \in L^n_j} l_j^n(g)\\
&\leq \sum_{i\in S}\frac{|A^*_i|-k}{|A^*_i|}(u_i(A^*_i)-u_{n}(A^*_i)).
\end{align*}
Therefore, our procedure outputs a cardinal allocation $\cA^k$ with utilitarian welfare
\begin{align*}
\USW(\cA^k) &\geq \sum_{i\in R\cup T}u_i(A^*_i)+\sum_{i\in S}\frac{k}{|A^*_i|}u_i(A^*_i)\\
&\quad +\sum_{i\in S}\frac{|A^*_i|-k}{|A^*_i|}u_{n}(A^*_i),
\end{align*}
    where $n\in R$.

\subsection{Proof of Lemma~\ref{lem:1catuppercase2} (Cont.)}
    \textbf{Case 2:} $\sum_{i\in S}|A^*_i|=m$.  In the remaining case where only agents in $S$ are allocated items under $\cA^*$, we have
     \begin{align*}
         \frac{\OPTUSW(I)}{\maxUSW(\mathcal{A})}&\leq \frac{\sum_{i\in S}u_i(A^*_i)}{1+\sum_{i\in S}\frac{k}{|A^*_i|}(u_i(A^*_i)-u_{i^\dagger}(A^*_i))}\\
         &=\frac{1+\sum_{i\in S}(u_i(A^*_i)-u_{i^\dagger}(A^*_i))}{1+\sum_{i\in S}\frac{k}{|A^*_i|}(u_i(A^*_i)-u_{i^\dagger}(A^*_i))}\\
         &=\frac{1+\sum_{i\in S}v_i(A^*_i)}{1+\sum_{i\in S}\frac{k}{|A^*_i|}v_i(A^*_i)},
     \end{align*}
     where $v_i(A^*_i)=u_i(A^*_i)-u_{i^\dagger}(A^*_i)$ for all $i\in S$.
     Similar to the proof of Lemma \ref{lem:fix}, we present the upper bound of the last expression through derivatives. Define $F(v_1(A^*_1),\ldots,v_{|S|}(A^*_{|S|})) = \frac{1+\sum_{i\in S}v_i(A^*_i)}{1+\sum_{i\in S}\frac{k}{|A^*_i|}v_i(A^*_i)}$ with domain $v_i(A^*_i) \in [0,1]$ for all $i$ and $\sum_{i\in S}v_i(A^*_i)\leq |S|-1$.

     Now for any $i'\in S$, the derivative of $F$ with respect to $v_{i'}(A^*_{i'})$ is
     $$
     \frac{\partial F}{\partial v_{i'}(A^*_{i'})}=\frac{{1-\frac{k}{|A^*_{i'}|}+k\sum_{i\in S}v_i(A^*_i)\left(\frac{1}{|A^*_i|}-\frac{1}{|A^*_{i'}|}\right)}}{{(1+\sum_{i\in S}\frac{k}{|A^*_i|}v_i(A^*_i))^2}}.
     $$
     Again, as in the proof of Lemma \ref{lem:fix}, the numerator is independent of $v_{i'}(A^*_{i'})$ and therefore the derivative is either negative, positive, or zero.
     Hence, $F$ is maximized when each $v_i(A^*_i)$ is either $0$ or is as large as possible, under the constraints of $\sum_{i\in S}v_i(A^*_i)\leq |S|-1$ and $u_i(A^*_i) \leq 1$. Let $E\subseteq S$ denote the set of indices such that $v_i(A^*_i)=0$ in the solution maximizing $F$. We split the proof into three further subcases based on whether or not $E=\emptyset$, and if $E=\emptyset$, whether or not each bundle of items belonging to the agents of $S$ under $\cA^*$ has the same number of items.
    
    If $E$ is empty \textbf{and} for every $i,j\in S$, $|A^*_i|=|A^*_j| = \frac{m}{|S|}$, then we have
    \begin{align*}
        \frac{\OPTUSW(I)}{\maxUSW(\mathcal{A})}&
        \leq  \max\limits_{\substack{0\leq v_i(A^*_i)\leq 1 \\ \sum_{i=1}^{|S|}v_i(A^*_i) \leq |S|-1}}F(\cdot)\\
        &\leq \frac{|S|}{1+k\frac{|S|}{m}(|S|-1)}\\
        &\leq \frac{\sqrt{\frac{m}{k}}}{2-\sqrt{\frac{k}{m}}}, 
    \end{align*}
    where the second inequality is a result of $|A^*_i|=|A^*_j|$ for all $i,j\in S$, as well as the requirement that $\sum_{i\in S}v_i(A^*_i)=|S|-1$ in order to maximize $F$.
    For the third inequality, some basic calculus shows that the expression is maximized when $|S| = \sqrt{\frac{m}{k}}$.
    It remains to show that for any $m>k\geq 1$, the following holds:
    $$
    \frac{\sqrt{\frac{m}{k}}}{2-\sqrt{\frac{k}{m}}} < \frac{1}{2}\left( \sqrt{1+\frac{m-1}{k}}+1 \right).
    $$
    Setting $r=\frac{m}{k}$, we have
    \begin{align*}
        \frac{1}{2}\Biggl( \sqrt{1+\frac{m-1}{k}}&+1 \Biggr)-\frac{\sqrt{\frac{m}{k}}}{2-\sqrt{\frac{k}{m}}}\\
        &=\frac{1}{2}\left( \sqrt{1+r-\frac{1}{k}}+1 \right)-\frac{\sqrt{r}}{2-\frac{1}{\sqrt{r}}}\\
        &\geq \frac{\sqrt{r}+1}{2}-\frac{r}{2\sqrt{r}-1}\\
        &=\frac{\sqrt{r}+1}{2}-\frac{r-\frac{\sqrt{r}}{2}}{2\sqrt{r}-1}-\frac{\sqrt{r}}{4\sqrt{r}-2}\\
        &=\frac{1}{2}-\frac{\sqrt{r}}{4\sqrt{r}-2}>0
    \end{align*}
    for all $r>1$. This gives us
    \begin{align*}
        \frac{\OPTUSW(I)}{\maxUSW(\mathcal{A})}&\leq \frac{\sqrt{\frac{m}{k}}}{2-\sqrt{\frac{k}{m}}}\\
        &<\frac{1}{2}\left( \sqrt{1+\frac{m-1}{k}}+1 \right),
    \end{align*}
    concluding the proof of this subcase.
    
    If $E$ is empty \textbf{and} $|A_p|\neq |A_q|$ for some $p,q\in S$, then suppose without loss of generality that $S=[|S|]$ and $|A_1|\geq |A_2|\geq \dots \geq |A_{|S|}|$.
    Recall that $\sum_{i=1}^{|S|}v_i(A^*_i) = |S|-1$ is a necessary condition for maximizing $F$.
    Then by the rearrangement inequality, $k\sum_{i\in S}\frac{v_i(A^*_i)}{|A^*_i|}$ is minimized when $v_i(A^*_i)=1$ for every $i\in [|S|-1]$, and $v_{|S|}(A^*_{|S|})=0$. 
    Since $|A_{|S|}|\geq k+1$, we have $\sum_{i=1}^{|S|-1}|A_i|\leq m-k-1$, so therefore
    \begin{align*}
        \frac{\OPTUSW(I)}{\maxUSW(\mathcal{A})}&\leq  \max\limits_{\substack{0\leq v_i(A^*_i)\leq 1 \\ \sum_{i=1}^{|S|}v_i(A^*_i) \leq |S|-1}}F(\cdot)\\
        &\leq 
        \frac{|S|}{1+\sum_{i\in [|S|-1]}\frac{k}{|A_i|}}\\
        &\leq
        \frac{|S|}{1+k\frac{(|S|-1)^2}{m-k-1}},
    \end{align*}
    where the last inequality transition is due to the AM-HM inequality.
    By computing the derivative of the last expression above with respect to $|S|$, one can verify that the expression is maximized when $|S|=\sqrt{\frac{m-1}{k}}$, and thus, the expression is at most
    
    $$\begin{aligned}
    &\frac{(m-k-1)\sqrt{\frac{m-1}{k}}}{m-k-1+k(\sqrt{\frac{m-1}{k}}-1)^2}\\
    &=\frac{(m-1)\left(\sqrt{\frac{m-1}{k}}+1\right)-\sqrt{k(m-1)}\left(\sqrt{\frac{m-1}{k}}+1\right)}{2(m-1)-2\sqrt{k(m-1)}}\\
        &=\frac{1}{2}\left(\sqrt{\frac{m-1}{k}}+1\right)< \frac{1}{2}\left(\sqrt{1+\frac{m-1}{k}}+1\right).\\
           \end{aligned}
    $$
    We have now completed the proof for the subcase where $E$ is empty. 
    
    If $E$ is non-empty, we define $S'\coloneqq S\setminus E$ and accordingly, we have
    \begin{align*}
   \frac{1+\sum_{i\in S}v_i(A^*_i)}{1+\sum_{i\in S}\frac{k}{|A^*_i|}v_i(A^*_i)}
   &\leq \frac{1+|S|-|E|}{1+\sum_{i\in S\setminus E}\frac{k}{|A^*_i|}}\\
   &= \frac{1+|S'|}{1+\sum_{i\in S'}\frac{k}{|A^*_i|}}.
    \end{align*}
     Since each agent $i\in E$ also belongs to $S$, we have $\sum_{i\in S'}|A_i^*| \leq m-(k+1) < m -1 $.
     Then by the AM-HM inequality, we have
     $$
   \sum_{i\in S'} \frac{1}{|A_i^*|} \geq \frac{|S'|^2}{\sum_{i\in S'} |A^*_i|} > \frac{|S'|^2}{m-1},
     $$
    implying 
    \begin{align*}
    \frac{\OPTUSW(I)}{\maxUSW(\mathcal{A})} &\leq \frac{1+|S'|}{1+k\frac{|S'|^2}{m-1}} \\
    &< \frac{1}{2}\left( \sqrt{1+\frac{m-1}{k}} + 1 \right),
    \end{align*}
    where the last inequality has been shown in the proof of Lemma \ref{lem:fix}.


\section{Omitted proofs for Section~\ref{sec::multi-cate}}
\subsection{Proof of Lemma~\ref{lem:red}}
    We prove this lemma statement via the following Lemmas~\ref{lem:red1}, \ref{lem:red2}, and \ref{lem:red3}.
\begin{lemma}\label{lem:red1}
Given an instance $I$ where some agent $i$ exceeds the cardinality constraint in more than one category under the utilitarian-optimal allocation $\cA^*$, there exists another instance $I'$ where under $\cA'^{*}$, agent $i$ exceeds the cardinality constraint in only one category, and $\frac{\OPTUSW(I)}{\maxUSW(\mathcal{A})}\leq \frac{\OPTUSW(I')}{\maxUSWa(\mathcal{A})}$ holds.
\end{lemma} 
\begin{proof}
For simplicity, we assume that only agent $1$ exceeds the cardinality constraint in the category subset $H$, where $|H| \geq 2$. Later on, we will explain how to extend the following proof to the case where both agents exceed the cardinality constraint in at least two categories each.
Denote by $\cA^k$ the utilitarian-optimal cardinal allocation and by $A^k_{ij}$ the bundle agent $i$ receives from $C_j$. 

Consider category $p\in \arg\min_j{\frac{k_j}{|A^*_{1j}|}}$, breaking ties arbitrarily. We consider another instance $I'$ that only differs from $I$ by the utility functions. 
In instance $I'$, each agent $i$'s utility function $u'_i$ is described as follows,
\begin{itemize}
    \item for any $g \notin \bigcup_{j\in H}A^*_{1j}$, $u'_i(g) = u_i(g)$;
    \item for any $g \in \bigcup_{j\in H\setminus \{p\}}A^*_{1j}$, $u'_i(g)=0$;
    \item for any $g \in A^*_{1p}$, $u'_i(g)=u_i(g)+\frac{\sum_{j\in H\setminus\{p\} }u_i(A^*_{1j})}{|A^*_{1p}|}$.
\end{itemize}
In other words, we merge each agent's utilities for all of the items which agent $1$ receives in the category $H$ into the items which agent $1$ receives in category $p$ such that the marginal increase of item utilities is identical.

Let $\cB^*$ (resp. $\cB^k$) be the utilitarian-optimal (resp. utilitarian-optimal cardinal) allocation of $I'$. 
We have $\USW(\cB^*)=\USW(\cA^*)$ and therefore, in order to prove that $\frac{\OPTUSW(I)}{\maxUSW(\mathcal{A})}\leq \frac{\OPTUSW(I')}{\maxUSWa(\mathcal{A})}$, it suffices to show that $\USW(\cB^k)\leq\USW(\cA^k)$.
For $\cA^k$ and $\cB^k$, if $j\notin H$, the total utility from the assignments of $C_j$ in these two allocations is equal. Then it suffices to prove that the total utility from the assignment of $\{C_j\}_{j\in H}$ in $\cB^k$ is at most that in $\cA^k$. 
For categories $\{C_j\}_{j\in H}$, the total utility from the goods $\bigcup_{j\in H} A^*_{2j}$ is identical between $\cA^k$ and $\cB^k$ as the agents' utilities are unchanged for these items in $I'$, compared to $I$. Accordingly, we next prove that the utility from $\bigcup_{j\in H} A^*_{1j}$ in $\cA^k$ is at least that in $\cB^k$.
By the construction of $u'$, only items in $A^*_{1p}$ (among $\bigcup_{j\in H} A^*_{1j}$) can result in non-zero utilities in $I'$.
Next we claim that the assignment of items in $ A^*_{1p}$ is identical in $\cA^k$ and $\cB^k$.

\begin{claim}\label{claim:multi-ak1}
    For any $i\in [2]$, $A^k_{ip} \cap A^*_{1p} = B^k_{ip} \cap A^*_{1p}$.
\end{claim}
\begin{proof}[Proof of Claim \ref{claim:multi-ak1}]
    It suffices to prove $A^k_{1p} \cap A^*_{1p} = B^k_{1p} \cap A^*_{1p}$ as we have already shown that the allocations of $A^*_{2p}$ are identical in $\cA^k$ and $\cB^k$. Furthermore, we know that $A^k_{1p}, B^k_{1p}\subseteq A^*_{1p}$, and hence, it suffices to prove $A^k_{1p} = B^k_{1p}$.

    First, as $u_1(A^*_{1j}) > u_2(A^*_{1j})$ for each $j\in H\setminus\{p\}$ and $u_1(g) \geq u_2(g) $ for every $g\in A^*_{1p}$, we have $u'_1(g) > u'_2(g)$ for all $g\in A^*_{1p}$.
    We then have $B^*_{1p} = A^*_{1p}$. Now let $S$ denote the set of items moved from agent $1$ to agent $2$, i.e., $S\coloneqq A^*_{1p}\setminus A^k_{1p}$.
    Then $S$ contains the $|A^*_{1p}|-k$ items whose reassignments cause the minimum utility loss under $u_i$'s. Note that by the construction of $u'$, we see that starting from $\cB^*$ and moving the items of $S$ from agent $1$ to agent $2$ also causes the minimum utility loss. Therefore, $B^k_{1p}=B^*_{1p}\setminus S = A^*_{1p}\setminus S = A^k_{1p}$, completing the proof.
\end{proof}

We now upper bound the total utility of the items in $A^*_{1p}$ for allocation $\cB^k$ as follows,
$$
\begin{aligned}
  &\sum_{i=1}^2  u'_i(B^k_{ip}\cap A^*_{1p})\\
  & = \sum_{i=1}^2u_i(A^k_{ip}\cap A^*_{1p}) \\
  &\quad + \sum_{j\in H\setminus\{p\}} \left( \frac{k_p}{|A^*_{1p}|}\cdot u_1(A^*_{1j})  + \frac{|A^*_{1p}|-k_p}{|A^*_{1p}|}\cdot u_2(A^*_{1j}) \right)  \\
  & = \sum_{i=1}^2u_i(A^k_{ip}\cap A^*_{1p})+ \sum_{ j \in H\setminus \{ p \} } u_2(A^*_{1j}) \\
  &\qquad + \frac{k_p}{|A^*_{1p}|}\sum_{j\in H\setminus\{p\}} (u_1(A^*_{1j}) - u_2(A^*_{1j}) )  \\
  & \leq \sum_{i=1}^2u_i(A^k_{ip}\cap A^*_{1p})+ \sum_{ j \in H\setminus \{ p \} } u_2(A^*_{1j}) \\
  &\qquad + \sum_{j\in H\setminus\{p\}} \frac{k_j}{|A^*_{1j}|}(u_1(A^*_{1j}) - u_2(A^*_{1j}) ) ,
\end{aligned}
$$
where the first equality is because $|A^k_{1p} \cap A^*_{1p}| = k_p$ and $|A^k_{2p}\cap A^*_{1p}| = |A^*_{1p}| - k$, and the inequality is due to $\frac{k_p}{|A^*_{1p}|} \leq \frac{k_j}{|A^*_{1j}|}$ for every $j \in H$.
By using similar arguments to that in the proof of Lemma \ref{lem:conscard}, one can verify that the last expression above is no greater than the utility from the assignment of $\bigcup_{j\in H}A^*_{1j}$ in allocation $\cA^k$.
Note that the first summation term is equal to the total utility from $A^*_{1p}$, and that the fact that the second summation term is at most $\sum_{j\in H\setminus\{p\}}\sum_{i=1}^2u_i(A^k_{ij}\cap A^*_{1j})$ can be checked by running the reassignment procedure in the proof of Lemma \ref{lem:conscard} for every $A^*_{1p}$ with $j\in H\setminus\{p \}$.
According to the construction of $I'$, with the exception of $\bigcup_{j\in H}A^*_{1j}$, the allocations of the other items in $\cA^k$ and $\cB^k$ are identical. As a consequence, we have $\USW(\cB^k) \leq \USW(\cA^k)$, and thus, the price of cardinality ratio of $I$ is no greater than that of $I'$.

For the case where both agents exceed the constraint in at least two categories, we can then repeat this process for agent $2$; note that the allocations of items in the category where agent $2$ exceeds the cardinality constraint are identical in $\cA^*$ (resp. $\cA^k$) and $\cB^*$ (resp. $\cB^k$) so that one can start from $I'$ and construct another $I''$ to merge the utility of items of agent $2$.
\end{proof}

\begin{lemma} \label{lem:red2}
Given an instance $I$ where under $\cA^{*}$, both agents exceed the cardinality constraint in at most one category each, there exists another instance $I'$ where under $\cA'^{*}$, if an agent exceeds the cardinality constraint in some category $j$, then she receives zero utility from every other category $[h]\setminus \{j\}$, and $\frac{\OPTUSW(I)}{\maxUSW(\mathcal{A})}\leq \frac{\OPTUSW(I')}{\maxUSWa(\mathcal{A})}$ holds.
\end{lemma}
\begin{proof}
Suppose without loss of generality that in $\cA^*$, agent $1$ exceeds the cardinality constraint of $C_1$ and also receives non-zero utility from $\{C_2, C_3,\ldots\}$.
By merging each agent's utility for items $A^*_{12}$, we construct an instance $I'$ that only differs from $I$ in the utility functions. Specifically, the utility function $u'_i$ of agent $i$ in $I'$ is described as follows;
\begin{itemize}
    \item for any $g\in A^*_{12}$, $u'_i(g)=0$;
    \item for any $g\in A^*_{11}$, $u'_i(g)=u_i(g)+\frac{1}{|A^*_{11}|}\cdot u_i(A^*_{12})$;
    \item for any remaining $g$, $u'_i(g)=u_i(g)$.
\end{itemize}
Let $\cB^*$ and $\cB^k$ denote the utilitarian-optimal allocation and the utilitarian-optimal cardinal allocation, respectively.
By comparing the utility profiles between $I$ and $I'$, we see that $\USW(\cA^*) = \USW(\cB^*)$, and thus it suffices to prove $\USW(\cA^k) \geq \USW(\cB^k)$. By arguments similar to that in the proof of Claim \ref{claim:multi-ak1}, one can verify that in allocations $\cA^k$ and $\cB^k$, the assignments of items $A^*_{11}$ are identical.
We now consider the total utility from items in $A^*_{11} $ in allocation $\cB^k$,
$$
\begin{aligned}
\sum_{i=1}^2 u'_i(B^k_{i1}\cap A^*_{11}) & = \sum_{i=1}^2 u_i(B^k_{i1}\cap A^*_{11}) + \frac{k_1}{|A^*_{11}|}u_1(A^*_{12})\\
&\quad + \frac{|A^*_{11}| - k_1}{|A^*_{11}|}u_2(A^*_{12}) \\
& = \sum_{i=1}^2 u_i(B^k_{i1}\cap A^*_{11})+ u_2(A^*_{12})\\
&\quad + \frac{k_1}{|A^*_{11}|} \left( u_1(A^*_{12}) - u_2(A^*_{12}) \right) .
\end{aligned}
$$
We now show that the last expression above is no greater than the utility from items $A^*_{11}\cup A^*_{12}$ in $\cA^k$.
The first summation term is equal to $u_1(A^*_{11})$, which is the utility from items $A^*_{11}$ in $\cA^k$. The sum of the last two terms is no greater than $u_1(A^*_{12})$ 
as $k_1<|A^*_{11}|$ and $u_1(A^*_{12}) \geq u_2(A^*_{12})$.
By the construction of $I'$, the assignments of items in $\cA^k$ and $\cB^k$ are identical, with the exception of $A^*_{11}\cup A^*_{12}$. Thus, we have $\USW(\cB^k) \leq \USW(\cA^k)$ and therefore $\frac{\OPTUSW(I)}{\maxUSW(\mathcal{A})}\leq \frac{\OPTUSW(I')}{\maxUSWa(\mathcal{A})}$.
We can repeat this process for both agents, and for each of the categories where an agent does not exceed the cardinality constraint but receives non-zero utility in $\cA^*$, concluding the proof.
\end{proof}

\begin{lemma}\label{lem:red3}
Given an instance $I$ where under $\cA^{*}$, at most one agent exceeds the cardinality constraint in some category, there exists another instance $I'$ where under $\cA'^{*}$, both agents exceed the cardinality constraint in one category each, and these categories are different. Also, $\frac{\OPTUSW(I)}{\maxUSW(\mathcal{A})}\leq \frac{\OPTUSW(I')}{\maxUSWa(\mathcal{A})}$ holds.
\end{lemma}
\begin{proof}
We ignore the case where neither agent violates the cardinality constraints, as it holds that
$\frac{\OPTUSW(I)}{\maxUSW(\mathcal{A})}=1$.
If only one agent exceeds the cardinality constraints, by Lemma~\ref{lem:red1}, it suffices to assume (w.l.o.g.) that agent $1$ exceeds (only) the constraint of category $1$. Then by Lemma~\ref{lem:red2}, we further assume that agent $1$ receives zero utility from $\{C_2, C_3,\ldots\}$.
As agent $2$ does not exceed the cardinality constraint of category $2$, agent $1$ receives at least $|A^*_{12}| \geq |C_2|-k_2$ items from $C_2$ and by our assumption, has a value of $0$ for each of them.
Since agents' utilities are normalized, agent $2$ must receive at least one non-zero utility item $g^*\in \cA^*$. We can construct another instance $I'$ where agent $2$'s utility for $g^*\in \cA^*$ is decreased by an arbitrarily small number $\epsilon>0$, and her utility for each item in $A^*_{12}$ is increased by $\frac{\epsilon}{|A^*_{12}|}$. As a result, agent $2$ now exceeds the cardinality constraint in category $2$.
The optimal utilitarian welfare of $I'$ is equal to $\USW(\cA^*)$, while that of utilitarian-optimal cardinal allocation slightly decreased as agent $2$ can only receive $k_2$ items in cardinal allocations. Therefore, $\frac{\OPTUSW(I)}{\maxUSW(\mathcal{A})}\leq \frac{\OPTUSW(I')}{\maxUSWa(\mathcal{A})}$.
\end{proof}

\end{document}